\documentclass[11pt,oneside,reqno]{amsart}

\usepackage[final]{graphicx}
\usepackage{amsfonts}
\usepackage{pdfsync}
\usepackage{amsmath}
\usepackage{amssymb}
\usepackage{amsthm}
\usepackage{caption}
\usepackage{enumerate} 
\usepackage{dcolumn} 

\newtheorem{theorem}{Theorem}[section]
\newtheorem{lemma}[theorem]{Lemma}

\newtheorem{proposition}[theorem]{Proposition}

\newcommand{\R}{{\mathord{\mathbb R}}}

\newcommand{\Z}{{\mathord{\mathbb Z}}}

\newcommand{\supp}{{\mathop{\rm supp\ }}}

\usepackage{color}

\newcommand\numberthis{\addtocounter{equation}{1}\tag{\theequation}}

\newcommand{\s}{\mathbb{S}}

\newcommand{\di}{\,\mathrm{d}}

\newcommand{\abs}[1]{\left|#1\right|}
\newcommand{\norm}[1]{\lVert#1\rVert}
\newcommand{\bvec}[1]{\boldsymbol{#1}}

\newcommand{\inprodtwo}[2]{\left \langle #1 , #2\right \rangle}

\begin{document}
\title[Minimizers for Poincar\'e--Sobolev inequalities]{\textbf{Existence and non-existence of minimizers for Poincar\'e--Sobolev inequalities}}

\author[Benguria]{Rafael~D.~Benguria$^1$}

\author[Vallejos]{Cristobal~Vallejos$^2$}

\author[Van Den Bosch]{Hanne~Van~Den~Bosch$^3$}

\address{$^1$ Instituto  de F\'\i sica, Pontificia Universidad Cat\' olica de Chile,}
\email{{rbenguri@fis.puc.cl}}

\address{$^2$ Facultad de F\'\i sica, Pontificia Universidad Cat\' olica de Chile,}
\email{{civallejos@uc.cl}}

\address{$^3$ Centro de Modelamiento Matem\'atico, CMM, FCFM, Universidad de Chile}
\email{{hvdbosch@cmm.uchile.cl}}

\begin{abstract} 
In this paper we study the existence and non-existence of minimizers for a type of (critical) Poincar\'{e}--Sobolev inequalities. We show that minimizers do exist for smooth domains in $\R^d$, an  also for some polyhedral domains. On the other hand, we prove the non-existence of minimizers in the rectangular isosceles triangle in $\R^2$. 
\end{abstract}

\maketitle

\section{Introduction} \label{intro}

In this paper, we continue the study of a special type of Poincar\'{e}--Sobolev inequalities, which are extensions to the case of bounded domains of Gagliardo--Nirenberg--Sobolev inequalities. For a bounded domain $\Omega \subset \R^d$, we define
\begin{equation}
\label{eq:var-prob}
G(\Omega, d) = \inf \frac{\int_{\Omega} \abs{\nabla u}^2   \left(\int_{\Omega} u^2   \right)^{2/d}}{\int_{\Omega} {\vert{u-u_\Omega \vert}}^{2+4/d}   },
\end{equation}
with the infimum taken over functions $u \in H^1(\Omega)$ and
\[
u_\Omega = \frac{1}{\abs{\Omega}} \int_\Omega u
\]
is the average of $u$.
In our previous work \cite{BeVaVdb018}, the main result was a lower bound for $G(\Omega,d)$ in convex domains. 
It was also shown that for $d= 1$, no minimizers exist.
Here, we concentrate on the existence of minimizers for $d \ge 2$. We will see that existence or non-existence depend strongly on the shape and regularity of the domain $\Omega$. 

Our main results are the following.
\begin{itemize}
\item {\bf Existence} of minimizers for $C^3$-smooth domains in $\R^d$ for $d \ge 2$. 
\item {\bf Existence} of minimizers in elongated rectangles.
\item {\bf Existence} of minimizers in hypercubes in $\R^d$ for $d \ge 10$.
\item {\bf Non-existence} of minimizers in the isosceles rectangular triangle.
\end{itemize}
 
In \cite{BeVaVdb018}, we conjectured the non-existence of  minimizers for the square, but proving this remains an open problem. From the result for the triangle, we obtain that minimizers in the square, if they exist, are not symmetric with respect to the diagonal.

The inequality corresponding to \eqref{eq:var-prob} in the whole of $\R_d$ is the Gagliardo--Nirenberg--Sobolev inequality (also known as Moser's inequality)
\begin{equation}
\int_{\R^d} \abs{\nabla u}^2     \ge G(d) \left( \int_{\R^d} u^2    \right)^{-2/d}   \int_{\R^d} u^{2(1+2/d)}   ,
\label{GNSRDIntro}
\end{equation} 
where $G(d)$ is the sharp constant.
In this case, it is well known that minimizers exist and are unique up to translations, scalings and space dilations. 

The main tool to establish both existence and non-existence of minimizers is a treshold for the loss of compactness (in the spirit of Brezis and Lieb \cite[Section 4.B]{BrLi983}).
For smooth domains, loss of compactness can only be due to concentration on the boundary.

\begin{theorem} \label{thm:treshold-smooth}
For a bounded $C^2$-domain $\Omega \subset \R^d$ , define $G(\Omega, d)$ as in \eqref{eq:var-prob} and let $G(d)$ be the sharp constant in \eqref{GNSRDIntro}. 
Then
\[
G(\Omega, d) \le G(d) /2^{2/d}
\]
and if the inequality is strict, a minimizer exists.
\end{theorem}

If the domain is not smooth, loss of compactness can be due to concentration at corners or edges. For simplicity, we state this result only for planar domains.
\begin{theorem} \label{thm:treshold-corners}
Let $\Omega \subset \R^2$ be a bounded planar domain, piecewise $C^2$ and with finitely many corners of interior angles $0<\alpha_j \le 2\pi$, $j=1, \dots ,N$. Assume for simplicity that $\partial \Omega$ does not have self-intersections. Define $G(\Omega, d)$ as in \eqref{eq:var-prob} and let $G(d)$ be the sharp constant in \eqref{GNSRDIntro}. 
Then
\[
G(\Omega, 2) \le G(2) \frac{1}{2\pi} \min (\pi, \alpha_1, \cdots , \alpha_N)
\]
and if the inequality is strict, a minimizer exists.
\end{theorem}

In \cite{BeVaVdb018}, the analogue of Theorems~\ref{thm:treshold-smooth} and \ref{thm:treshold-corners}  was proven for the special case $\Omega = [0,1]^d$. Although the heuristic idea remains the same, the proof given there relied on a rearrangement inequality that is only valid in cubes or polygons. 
Here, we give a different proof using localization with a well-chosen partition of unity.

Theorems~\ref{thm:treshold-smooth}
and \ref{thm:treshold-corners} will be proven in Section~\ref{sec:treshold}. The proof of existence of minimizers for $C^3$-smooth domains in $d\ge 2 $ is contained in Section~\ref{sec:smooth}. 
Here, the idea is to construct competitors by concentrating the minimizer of the problem in $\R^d$ at a suitable boundary point and obtain the sign of the next to leading order in the expansion of the quotient \eqref{eq:var-prob}. This strategy goes back to the original work of Brezis-Nirenberg \cite{BrNi983}. Its use in the present context was suggested to us by Rupert Frank.
Contrary to the case of \cite{BrNi983}, in this paper we expand near a boundary point. Such expansions are common in the literature on Partial Differential Equations, see for instance \cite{PiFeWe999, Dip011} and references therein.

The proof of non-existence of minimizers for the rectangular isosceles triangle, based on Theorem~\ref{thm:treshold-corners} and symmetry considerations, is in Section~\ref{sec:triangle}. The final section~\ref{sec:rectangles} contains the proofs of existence for rectangles in $d = 2$ and hypercubes in $d \ge 10 $.

 The examples we give show that existence or non-existence of minimizers depend in a non-trivial way on the geometry of the boundary of the domain. This is because the problem is precisely scale invariant. In the appendix, we show explicitly that, for the generalized problem 
\begin{equation*}
G_p(\Omega, d) = \inf_{H^1(\Omega)} \frac{\int_\Omega \abs{\nabla u}^2 \left(\int_\Omega u^2 \right)^{p}}{\int_\Omega \abs{u-u_\Omega}^{2+2p}}, \quad 0 \le p \le 2/(d-2) 
\end{equation*}
minimizers exist for $p < 2/d$ and do not exist for $p > 2/d$.

\section{Compactness treshold} \label{sec:treshold}

In this section, we prove Theorems~\ref{thm:treshold-smooth} and \ref{thm:treshold-corners}. Before going into the details, let us quickly sketch the philosophy. The upper bound is easy by constructing a sequence of test functions consisting of the minimizer of the problem in $\R^d$ concentrating at a boundary point (respectively the corner of smallest opening). 

Then, we prove that non-existence of a minimizer implies the reverse inequality. 
In order to do so, we observe that non-existence can only be due to concentration of minimizing sequence. We localize the concentrating sequence at a suitable scale and pass to the \emph{model problem} on a cone by straightening the boundary. For smooth domains, all model problems are the same and give the constant for the halfspace. For the curvilinear polygon, some points give different constants, but the smallest one is given by the smallest angle. 

The following lemma takes care of the localization, which does not require regularity of the boundary or the specific exponent $p=2/d$.
\begin{lemma}[Localization] \label{lem:localizing}
 Fix $\delta > \eta >0$, and $p \in (0,2/(d-2))$. For all $v \in H^1(\Omega)$ with $\norm{v}_{L^2} = 1$, 
 we have
 \begin{align} \label{eq:model problem}
  \int_\Omega \abs{\nabla v}^2  \ge G_p(\Omega, d,2\sqrt{d}\delta)(1-C \eta \delta^{-1}) \int_\Omega \abs{v}^{2+p} -C \eta^{-2},
 \end{align}
where $C$ is a constant depending only on $p, d$ and
\begin{equation}
 G_p(\Omega, d,\delta) = \min_{s \in \R^d} \inf_{v \in H_0^1(B(s,\delta))} \frac{\int_\Omega \abs{\nabla v}^2 \left(\int_\Omega v^2 \right)^p}{\int_\Omega \abs{v}^{2+2p}},
\end{equation}
with the convention that the quotient equals $+ \infty$ if the denominator equals zero.
\end{lemma}

\begin{proof}
 We localize in cubes of size $\delta$ with a smooth cut-off varying on lengths $\eta$.
Explicitly, we make the following construction. Fix a smooth non-increasing function $\phi: \R \mapsto [0,1]$ such that $\phi(x) = 1$ for $x \le -1/2$ and $\phi(x) = 0$ for $x \ge 1/2$, and such that $\phi(x)^2 + \phi(-x)^2 = 1$.
We define
\[
\chi (x) = \prod_{i= 1}^d \phi \left(\frac{\abs{x_i}- \delta}{\eta} \right).
\]
By induction on $d$, one can show that, since $\delta > \eta$ 
\[
1(x) = \sum_{k \in \Z^d} \chi^2 (x- 2 \delta k).
\]
For $k \in \Z^d$, we define $\chi_k (x)= \chi (x-2 \delta k) $ and $v_{k} = \chi_k v $.
From this point on, $C$ will denote a constant depending on the choice of $\phi$, and the dimension $d$.

By the IMS formula (see \cite[Theorem 3.1]{CFKSbook} or the original research papers \cite{Ism961, Mor979,MorSi980, Sig982}) , we have 
\begin{align*}
\int_\Omega \abs{\nabla v} 
&\ge \sum_{k \in \Z^d} \int_{\Omega \cap \supp (\chi_k) }\left( \abs{\nabla v_{k}}^2 - \frac{C}{\eta^2} v^2 \right)\\
& \ge \sum_{k \in \Z^d} \int_\Omega  \abs{\nabla v_{k}}^2 - \frac{C}{\eta^2} 
\end{align*}
where we have bounded $\abs{\nabla \chi} \le C \eta^{-1}$ and used the fact that a fixed point $x$ is in the support of at most $2^d$ cut-off functions, and finally the normalization of $v$
By construction, $\supp v_k \subset B(\delta k, 2\sqrt{d} \delta )$, so
\begin{align*}
\int_\Omega \abs{\nabla v_{k}}^2 \left(\int_\Omega v_{k}^2\right)^{p} \ge G_p(\Omega, d, 2\sqrt{d}\delta) \int_\Omega \abs{v_k}^{2+2p}. 
\end{align*}

Combining with
\[
1 = \left(\int_\Omega v^2\right)^{p} \ge \left(\int_\Omega v_{k}^2\right)^{p}
\]
gives
\begin{align*}
\int_\Omega\abs{ \nabla v}^2
& \ge  \sum_{k \in \Z^d} \left(\int_\Omega  \abs{\nabla v_{k}}^2 \left( \int_\Omega v_{k}^2\right)^{p} \right)- \frac{C}{\eta^2} \\
& \ge G_p(\Omega, d, 2\sqrt{d}\delta ) \left[ \int_{\Omega} \abs{v}^{2+2p} -
\int_\Omega \abs{v}^{2+2p}\sum_{k \in \Z^d} \left( \chi_k^2-\chi_k^{2+2p}   \right)\right] -  \frac{C}{\eta^2}. \\
\end{align*}

In order to bound the first error term, we average over the position of the origin in $[-\delta, \delta]^d$, which corresponds to replacing $\chi_k(x)$ by $\chi_k(x- u)$.
We bound
\begin{align*}
&\frac{1}{(2 \delta)^d}\int_{[-\delta, \delta]^d}  \left (\int_\Omega \abs{v}^{2+2p}(x)\sum_{k \in \Z^d} \left( \chi_k^2(x-u)-\chi_k^{2+2p}(x-u)   \right)\di x \right) \di u\\
&= \frac{1}{(2 \delta)^d} \int_\Omega\abs{v}^{2+2p}(x) \left( \sum_{k \in \Z^d} \int_{[-\delta, \delta]^d} \left( \chi_k^2(x-u)-\chi_k^{2+2p}(x-u)   \right) \di u \right)  \di x \\
&\le \frac{2^d}{(2 \delta)^d} \int_\Omega\abs{v}^{2+2p}(x) \di x  \int_{\R^d} \left( \chi^2(y)-\chi^{2+2p}(y)   \right) \di y   \\
&\le C \frac{( \delta + \eta)^d - ( \delta - \eta)^d}{(2 \delta)^d} \int_\Omega\abs{v}^{2+2p}(x) \di x   \\
&\le C \frac{\eta}{\delta} \int_\Omega\abs{v}^{2+2p}(x) \di x .
\end{align*}
In the third line, we have used the fact that the point $x$ is in the support of at most $2^d$ localization functions to get rid of the sum over $k$ before changing variables. Finally, we use the fact that the support of $\chi $ is included in $[-\delta - \eta, \delta + \eta]^d$ and $\chi = 1$ in
$[-\delta + \eta, \delta - \eta]^d$.
This concludes the proof of the lemma.
\end{proof}

\begin{proof}[Proof of Theorem~\ref{thm:treshold-smooth}]

{\noindent \bf Upper bound.}
We center coordinates such that the origin is at the boundary and $\Omega$ is contained in the halfspace $\R^{d}_+$. 
For $\lambda \ge 1$, we define
$$
u_\lambda (x) = \lambda^{d/2} g( \lambda x), 
$$
where $g$ is a minimizer for \eqref{GNSRDIntro} centered at the origin. 
By scaling and radial symmetry, we find
\begin{align*}
0 &\le \abs{\Omega}^{-1}\int_{\Omega} u_\lambda    \equiv \overline{u_\lambda} \le C\lambda^{-d/2}   ,\\
&\int_{\Omega} u_\lambda^2  \le  \frac{1}{2} \int_{\R^d} g^2,\\
&\int_{\Omega} \abs{\nabla u_\lambda}^2   \le \lambda^2  \frac{1}{2} \int_{\R^d} \abs{\nabla g}^2 .
\end{align*}
For the denominator, we first use convexity of $t \mapsto t^{2+4/d}$,
\begin{align*}
\int_\Omega \abs{u_\lambda-\overline{u_\lambda}} ^{2+4/d}  
& \ge \int_\Omega u_\lambda^{2+4/d} - C \,\overline{u_\lambda} \int_\Omega u_\lambda^{1+4/d } \\
& \ge \int_\Omega u_\lambda^{2+4/d} - C \lambda^{-d/2} \lambda^{2-d/2}
\end{align*}
We fix $r_1 >0$ such that $\Omega \cap B(0,r_1)$ is the epigraph of a $C^2$-function $h : \R^{d-1} \cap B(0,r_1) \mapsto \R_+$, see Figure~\ref{fig:coordinate transforms}.
By assumption, $h$ vanishes to second order at $0$. Upon taking a smaller $r_1$, we may also assume that $h \le  r_1/2$.
Then, we find
\begin{align*}
 \int_{\Omega \cap B(0, r_1)} u_\lambda^{2+4/d} \ge \lambda^{2+d} \int_{\R^d_+ \cap B(0, r_1/2)} g^{2+4/d}(\lambda (x_1+h (x_t)), \lambda x_t) \di x_1 \di x_t\\
 = \lambda^2 \int_{\R^d_+ \cap B(0, \lambda r_1/2)} g^{2+4/d}(x_1+ \lambda h (x_t/ \lambda)), x_t) \di x_1 \di x_t \\
 \ge \frac{ \lambda^2}{2} \int_{B(0,\lambda r_1/2)} g^{2+4/d} - C \lambda \sup_{x_t \in B(0,  r_1/2)} \left(\abs{x_t}^{-2} h(x_t) \right) \int_{\R^d} \abs{x}^2 \abs{g'}(x) g^{1+4/d}(x) \di x,
\end{align*}
where use the convention $\abs{g'} = \abs{\nabla g} = \abs{\partial_{\abs{x}} g}$.
By the $C^2$-regularity of the boundary,
\[
 \sup_{x_t \in B(0,  r_1/2)} \abs{x_t}^{-2} h(x_t) \le C . 
\]
Using the exponential decay of $g$ we can bound
\begin{align*}
\int_\Omega \abs{u_\lambda-\overline{u_\lambda}} ^{2+4/d}  
& \ge  \frac{ \lambda^2}{2} \int_{\R^d} g^{2+4/d} - C \lambda .
\end{align*}
Thus, we obtain
\begin{align*}
G(\Omega, d) 
&\le \liminf_{\lambda \to \infty } \frac{\int_\Omega \abs{\nabla u_\lambda}^2    \left(\int_\Omega u_\lambda^2    \right)^{2/d}}{\int_\Omega \abs{u_\lambda-\overline{u_\lambda}} ^{2+4/d}   } \\
& \le 2^{-2/d} \frac{\int_{\R^d} \abs{\nabla g}^2    \left(\int_{\R^d} g^2    \right)^{2/d}}{\int_{\R^d} g^{2+4/d}  } = 2^{-2/d} G(d).
\end{align*}

\begin{figure}
\includegraphics[scale=1]{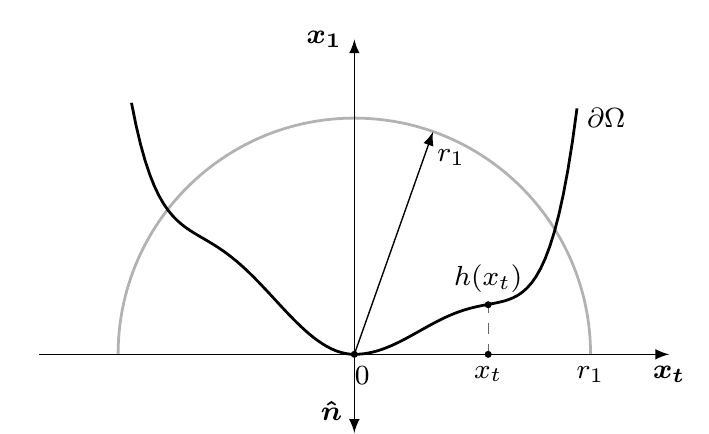} \hfill
\includegraphics[scale=1]{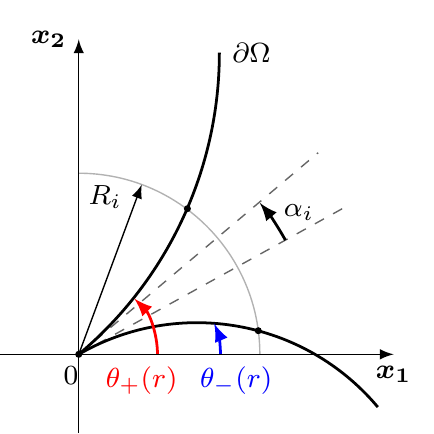}
\caption{\label{fig:coordinate transforms}}
Definition of coordinate transforms to straighten the boundary in the smooth case (left), or to map the boundary to a straight cone (right).
\end{figure}

 {\noindent \bf Lower bound.}
Let $u_n$ be a minimizing sequence for $\eqref{eq:var-prob}$, normalized such that $\int_\Omega u_n^2 = 1$ and $\int_\Omega u_n = 0$. By a standard argument, a minimizer exists if $\int_\Omega u_n^{2+4/d}$ is bounded along some subsequence. Therefore, we may assume that
\[
m_n = \norm{u_n}_{L^{2+4/d}(\Omega)} \to + \infty.
\]
For each $n$, we apply Lemma~\ref{lem:localizing} with $v= u_n$, $\delta= m_n^{-1/2}$ and $\eta = m_n^{-1}$.
With this choice, the lemma becomes
\begin{align*}
 G(\Omega, d)& = \lim_{n \to \infty} \frac{\int_\Omega \abs{\nabla u_n}^2}{\int_\Omega \abs{u_n}^{2+4/d}} \\
 &\ge \lim_{n \to \infty} G_{4/d}(\Omega, d, 2 \sqrt{d} \,m_n^{-1/2}) (1- C m_n^{-1/2}) - C m_n^{2-2-4/d} \\
& = \limsup_{\delta \to 0}  G_{4/d}(\Omega, d, \delta).
\end{align*}

If $s\in \R^d$ is such that $B(s, \delta) \subset \Omega$, we identify $v \in H_0^1(B(s,\delta))$ with its extension by $0$ in $H^1(\R^d)$, so
 \begin{equation}
  \inf_{v \in H_0^1(B(s,\delta))} \frac{\int_\Omega \abs{\nabla v}^2 \left(\int_\Omega v^2 \right)^{2/d}}{\int_\Omega \abs{v}^{2+4/d}} \ge G(d).
 \end{equation}
On the other hand, if $B(s,\delta)$ intersects the boundary of $\Omega$, we may as well replace $\delta$ by $2 \delta$ and assume $s \in \partial \Omega$. We assume that $\delta$ is sufficiently small such that, for each $s \in \partial \Omega$, the boundary $\partial \Omega$ can be seen as the graph of a $C^2$-function over the tangent plane. 

For definiteness, we fix a coordinate system with the origin at $s$ and the outward normal pointing along $-e_1$. Define $h$ as before.
For $v \in H^1_0(B(s,\delta))$, we define $f \in H^1(\R_+^d)$ by 
\[
f(x_1, x_t) = v(x_1+ h(x_t), x_t).
\] 
We compute
\begin{align}
\int_{\R^d_+}f^2 &= \int_\Omega v^2 ,  \qquad   
\int_{\R^d_+}f^{2+4/d} = \int_\Omega v^{2+4/d},  \quad \text{ and }  \label{eq:straightening-jac}\\
\int_{\R^d_+}\abs{\nabla f}^2 &\le \int_\Omega \abs{\nabla v}^2 (1 + \abs{\nabla_t h})^2 
\le \sup_{x_t \in B(s, \delta )}(1 + \abs{\nabla_t h(x_t)})^2 \int_\Omega \abs{\nabla v}^2. \label{eq:straightening-grad}
\end{align}
Since $\partial \Omega$ is $C^2$, $h$ is a $C^2$-function and
$\abs{\nabla_t h} \le C \delta $.
Since functions on the halfspace can be extended to $\R^d$ by reflection,
\begin{align*}
\frac{\int_\Omega \abs{\nabla v}^2 \left(\int_\Omega v^2\right)^{2/d}}{\int_\Omega v^{2+4/d}} 
&\ge (1 + C \delta)^{-2} \frac{\int_{\R^d_+}\abs{\nabla f}^2  \left( \int_{\R^d_+}\abs{\nabla f}^2\right)^{2/d}}{ \int_{\R^d_+}f^{2+4/d}}  \\
&\ge (1 + C \delta)^{-2} 2^{-d/2}G( d) .\numberthis \label{eq:straightening-conclusion}
\end{align*}
In summary, for sufficiently small $\delta > 0$,
\begin{align*}
G_{4/d}(\Omega, d, \delta) \ge (1 + C \delta)^{-2} 2^{-d/2}G( d) ,
\end{align*}
and, if a minimizer does not exist,
\begin{align*}
 G(\Omega, d)
 & \ge 
\limsup_{\delta \to 0}  G_{4/d}(\Omega, d, \delta) \ge 2^{-d/2}G( d). \qedhere
\end{align*}

\end{proof}

Now, we prove Theorem~\ref{thm:treshold-corners}. The proof is very similar, so we will sketch it and point out the differences due to the corners.

\begin{proof}[Proof of Theorem~\ref{thm:treshold-corners}]
{\bf Upper bound.}
If $\min(\pi, \alpha_1, \cdots \alpha_n) = \pi$, the bound holds by concentrating the minimizer of the problem in the plane, on one of the smooth points of the boundary.
If not, we take the origin at the vertex of some corner of opening $\alpha_i$. By assumption, there is $R_i>0$ such that $B(0,R_i)\cap \Omega $ contains no other corners and for each $r\le R_i$, $\partial B(0,r)\cap \Omega $ is simply connected.
In polar coordinates, $B(0,R_i)\cap \partial \Omega $ is given by $C^2$-functions $2\pi \ge \theta_+ (r)> \theta_-(r) \ge 0  $, as illustrated in Figure~\ref{fig:coordinate transforms}, left panel.
We have
\begin{equation}\label{eq:prop-theta}
\abs{\frac{\theta_+(r) - \theta_-(r)}{\alpha_i} - 1} \le C r, \quad \lim_{r \to 0} r \theta_{\pm}(r) = 0.
\end{equation}
Let $g$ be a minimizer of \eqref{GNSRDIntro}, $0\le\chi\le 1$ a smooth radial cut-off function with support in $B(0,R_1)$, $\chi = 1$ in $B(0,R_1/2)$, and define
\[
u_\lambda(x) = \chi(x) \lambda g(\lambda x).
\]
For the average,
\begin{align*}
0 &\le \abs{\Omega}^{-1}\int_{\Omega} u_\lambda    \equiv \overline{u_\lambda} \le C\lambda^{-1}.  
\end{align*}
For the other integrals, we use polar coordinates and \eqref{eq:prop-theta},
\begin{align*}
\int_{\Omega} u_\lambda^2 &\le \lambda^{2} \int_0^{R_1} (\theta_+(r) - \theta_-(r))  g^2(\lambda r) r \di r \\
&\le \frac{\alpha_i }{2 \pi} \int_{B(0,\lambda R_1)} g^2(x) \di x  + C \lambda^{-1} \int_{B(0,\lambda R_1)} \abs{x} g^2(x) \di x 
\end{align*}
and (recall that $u_\lambda$ is a radial function)
\begin{align*}
\int_{\Omega} \abs{\nabla u_\lambda}^2   
&\le \lambda^2\int_0^{R_1} (\theta_+(r) - \theta_-(r))  (\lambda \chi g' + \chi' g)^2 (\lambda r) r \di r\\
& \le \frac{\alpha_i }{2 \pi}\lambda^2 \int_{B(0,\lambda R_1)} \abs{\nabla g}^2(x) \di x + C\lambda \int_{B(0,\lambda R_1)} \left( \abs{x} (g')^2(x) + |gg'|(x)\right) \di x .
\end{align*}
Finally, for the denominator,
\begin{align*}
\int_{\Omega} \abs{u_\lambda - \overline{u_\lambda}}^4 \ge \frac{\alpha_i }{2 \pi}\lambda^2 \int_{B(0,\lambda R_1)} g^4(x) \di x - C \lambda \int_{B(0,\lambda R_1)} \abs{x} g^4(x) \di x - C.
\end{align*}
Since $g$ decays exponentially, we obtain
\begin{align*}
G(\Omega, d) 
&\le \liminf_{\lambda \to \infty } \frac{\int_\Omega \abs{\nabla u_\lambda}^2    \left(\int_\Omega u_\lambda^2    \right)^{2/d}}{\int_\Omega \abs{u_\lambda-\overline{u_\lambda}} ^{2+4/d}   } \\
& \le \frac{\alpha_i }{2 \pi} \frac{\int_{\R^d} \abs{\nabla g}^2    \left(\int_{\R^d} g^2    \right)^{2/d}}{\int_{\R^d} g^{2+4/d}  } = \frac{\alpha_i }{2 \pi} G(d).
\end{align*}

\noindent {\bf Lower bound.} 
Again, let $u_n$ be a minimizing sequence for $\eqref{eq:var-prob}$, normalized such that $\int_\Omega u_n^2 = 1$ and $\int_\Omega u_n = 0$ such that
\[
m_n = \norm{u_n}_{L^{2+4/d}(\Omega)} \to + \infty.
\]  
As in the proof of Theorem~\ref{thm:treshold-smooth}, by lemma~\ref{lem:localizing},
\begin{align*}
 G(\Omega, 2)& \ge
\limsup_{\delta \to 0}  G_{2}(\Omega, 2, \delta).
\end{align*}
We have to minimize the quotient in \eqref{eq:model problem}.
If $B(s,\delta)$ does not contain any corners, 
then \eqref{eq:straightening-conclusion} holds as before,
 \begin{align*}
\frac{\int_\Omega \abs{\nabla v}^2 \int_\Omega v^2}{\int_\Omega v^{4}} 
&\ge (1 + C \delta)^{-2} 2^{-1}G( 2),
\end{align*}
for some $C \ge 0$.

If $B(s, \delta)$ does contain corners, we may as well assume that $s$ is a corner of opening $\alpha_i$, and (up to taking a smaller $\delta$) that the boundary of $\Omega$ in $B(s, \delta)$ is described by $\theta_\pm(r)$ satisfying \eqref{eq:prop-theta} as above.
To leading order in $\delta$, the variational problem in $B(s, \delta)$ is equivalent to the problem on the circular sector 
$$C_{\alpha_i}= \{(r,\phi)| \phi \in (0,\alpha_i)\}.$$
Indeed, for $v \in H^1_0(B(s,\delta))$, we define $f \in H^1(C_{\alpha_i})$ by
\[
f(r,\phi) = v(r,\theta(r,\phi)), 
\quad \theta(r, \phi) = \frac{\theta_+(r) - \theta_-(r)}{\alpha_i} \phi + \theta_-(r). 
\]
This change of variables maps $C_{\alpha_i} \cap B(s,\delta)$ to $\Omega \cap B(s,\delta)$ and has Jacobian
\[
J(r) = \abs{\frac{\partial \phi}{\partial \theta}} = \frac{\theta_+(r) - \theta_-(r) }{\alpha_i} = 1 + O(r).
\]
For the gradient terms, we use \eqref{eq:prop-theta} again to bound
\begin{align*}
\abs{\nabla f}^2(r, \phi) &= \left( \partial_r v(r,\theta ) + \frac{\partial \theta}{\partial \phi}\partial_\theta v (r,\theta ) \right)^2 
+ \frac{1}{r^2}(\partial_\theta v)^2 \left(\frac{\theta_+(r) - \theta_-(r) }{\alpha_i}  \right)^2 \\
& \le \left( \partial_r v(r,\theta )  \right)^2 \left( 1 + r \frac{\partial \theta}{\partial \phi} \right) +  \frac{1}{r^2}(\partial_\theta v)^2 \left(1 + O(r)  +r \frac{\partial \theta}{\partial \phi}  + r^2 \left(\frac{\partial \theta}{\partial \phi} \right)^2\right) \\
& \le \abs{\nabla v}^2(r, \theta) (1 + O(r)),
\end{align*}
so we find
\[
\inf_{v \in H^1_0(B(s,\delta))} \frac{\int_\Omega \abs{\nabla v}^2 \int_{\Omega} v^2}{\int_\Omega v^4} \ge (1 - C \delta)\inf_{f \in H^1_0(B(s,\delta))} \frac{\int_{C_{\alpha_i}} \abs{\nabla f}^2 \int_{C_{\alpha_i}} f^2}{\int_{C_{\alpha_i}} f^4}. 
\]
Finally, within each cone $C_{\alpha_i}$, we use spherically decreasing rearrangements to show that it is equivalent to minimize over radial functions $f$. Minimization over radial problems is identical in $C_{\alpha_i}$ and in $\R^2$, so we find
\[
\inf_{f \in H^1_0(B(s,\delta))} \frac{\int_{C_{\alpha_i}} \abs{\nabla f}^2 \int_{C_{\alpha_i}} f^2}{\int_{C_{\alpha_i}} f^4} \ge \frac{\alpha_i}{2 \pi} G(2).
\]
Putting everything together, we have obtained
\begin{align*}
 G(\Omega, 2)& \ge
\limsup_{\delta \to 0}  G_{2}(\Omega, 2, \delta) \\
&\ge \limsup_{\delta \to 0} 
\frac{1}{2 \pi}\min\left(\pi, \alpha_1, \cdots, \alpha_N \right) G(2) (1 - C \delta) \\
&= \frac{1}{2 \pi}\min\left(\pi, \alpha_1, \cdots, \alpha_N \right) G(2),
\end{align*}
if a minimizer does not exist.
\end{proof}

\section{Smooth domains} \label{sec:smooth}
\noindent In this section, we prove
\begin{theorem} \label{thm:existence-smooth}
 Let $\Omega \subset \R^d$ be a $C^3$-domain and $d \ge 2$. Then a minimizer for $G(\Omega, d)$ given by \eqref{eq:var-prob} exists and 
 \[
G (\Omega, d) < G(\R^{d}_+, d) = G(d)/2^{2/d},
 \]
 where $\R^d_+$ is the halfspace $\R_+ \times \R^{d-1}$.
\end{theorem}

By Theorem~\ref{thm:treshold-smooth}, for this, it is sufficient to construct a competitor that makes the quotient in \eqref{eq:var-prob} smaller than $2^{-2/d}G(d)$.
It turns out that this is always possible in the smooth case by concentrating the minimizer of the problem in $\R^d$ at a boundary point with positive mean curvature. We are grateful to Rupert Frank for pointing out this idea to us. 
By the previous argument, the leading order for a sequence of test functions concentrating at any boundary point will give $2^{-2/d}G(d)$. In order to capture the next-to-leading order, we need to assume some additional regularity of the boundary.
We need the following well-known result from differential geometry.

\begin{lemma} \label{lem:convexity}
 Let $\Omega \subset \R^d$, $d \ge 2$ be a bounded $C^3$-domain. Then $\partial \Omega$ has at least one point where all the principal curvatures are non-negative and the mean curvature is strictly positive.
\end{lemma}

\begin{proof}[Proof of Theorem~\ref{thm:existence-smooth}]
Fix coordinates such that the origin coincides with a point of the boundary with non-negative curvatures given by Lemma~\ref{lem:convexity} and rotate the axis such that the outward normal at the origin coincides with $-e_1$.

Let $f$ be the radially decreasing minimizer of \eqref{GNSRDIntro} scaled to satisfy
\begin{equation}
\label{eq:EL-eq-Rd}
-\Delta f + f - f \abs{f}^{4/d} = 0.
\end{equation}
We define $f_\epsilon (x) = f(x/ \epsilon)$, with $\epsilon \le 1$. Our goal is to show that
\begin{align*}
\frac{\int_\Omega\abs{ \nabla f_\epsilon}^2 \left( \int_\Omega f_\epsilon^2 \right)^{2/d}}{\int_{\Omega} \abs{f_\epsilon - \overline{ f_\epsilon}}^{2+4/d}} 
\le \frac{G(d)}{2^{2/d}} \left( 1 - \epsilon C_\Omega \right) + O(\epsilon^{1+ \delta})
\end{align*}
for some $C_\Omega>0$ and $\delta > 0$. 
This implies that $G(\Omega, d) < {G(d)}/{2^{2/d}} $ and thus that a minimizer exists.

For some sufficiently small $R_1>0$, the surface $\partial \Omega \cap B(0,R_1)$ is the graph of a function of the form, 
\[
(x_1, x_t) \in \partial \Omega \cap B(0,R_1) \Rightarrow x_1 = \frac{1}{2}\inprodtwo{x_t}{ \bvec K x_t} + O(x_t^3)
\]
where $\bvec K$ is a matrix with the principal curvatures at the origin as eigenvalues. By assumption, $\bvec K$ is positive semidefinite, and at least one eigenvalue is positive. 

In addition, $f$ is exponentially decreasing. A standard application of the maximum principle gives that, for all $\mu < 1$, there exists $M > 0$ such that
\begin{align*}
f(x)+ \abs{\nabla f(x)} \le M e^{-\mu \abs{x}}, \text{ for all } \abs{x} \ge R_1.  
\end{align*}

Now we bound the quotient in \eqref{eq:var-prob}. For the $L^2$-norm we find
\begin{align*}
\int_\Omega f_\epsilon^2 = \int_{\Omega \cap B(0,R_1)} f_\epsilon^2 + \int_{\Omega \setminus B(0,R_1)}f_\epsilon^2 \le \epsilon^d \int_{\Omega_\epsilon} f^2 +  C e^{-\mu R_1 / \epsilon},
\end{align*}
where we have defined the scaled domain
\[
\Omega_\epsilon \equiv \epsilon^{-1} \left( \Omega \cap B(0,R_1) \right).
\]
Analogously, for the gradient term we have
\begin{align*}
\int_\Omega \abs{\nabla f_\epsilon} ^2 &= \int_{\Omega \cap B(0,R_1)} \abs{\nabla f_\epsilon} ^2 + \int_{\Omega \setminus B(0,R_1)} \abs{\nabla f_\epsilon} ^2 \\
&\le \epsilon^{d-2} \int_{\Omega_\epsilon} \abs{\nabla f} ^2 + C e^{-\mu R_1/ \epsilon}.
\end{align*}
For the average, we find
\[
\overline{ f_\epsilon} = \frac{\epsilon^d}{\abs{\Omega}} \int_{\epsilon x \in \Omega} f(x) \di x \le C \epsilon^d ,
\]
so we obtain
\begin{align*}
\int_{\Omega} \abs{f_\epsilon - \overline{ f_\epsilon}}^{2+4/d}
&\ge \int_{\Omega} f_\epsilon^{2+4/d} -(2+4/d) \overline{ f_\epsilon} \int_{\Omega} f_\epsilon^{1+4/d} \\
&\ge \epsilon^d\left(\int_{\Omega_\epsilon} f^{2+4/d} - C \epsilon^d  \right).
\end{align*}

Now we need to estimate integrals of positive radial functions over the domains $\Omega_\epsilon$. We will show below that for radial, nonnegative functions $g$,
\begin{equation} \label{eq:angular-integral}
\int_{\Omega_\epsilon} g(x) \di x = \frac{1}{2} \int_{\R^d} g(x) \di x - \epsilon \kappa C_d \int_{\R^d} g(x) \abs{x}\di x + R,
\end{equation}
with $C_d > 0$ depending only on the dimension and $\kappa>0$ the mean curvature at the origin. The error term can be bounded by
\[
\abs{R} \leq C \epsilon^2 \int_{B(0,R_1/\epsilon)} \left(\abs{g(x)} \abs{x}^2 + \abs{g'(x)} \abs{x}^3 \right) \di x + \int_{\R^d \setminus B(0, R_1/\epsilon) }\abs{g(x)} \di x.
\]

Assuming \eqref{eq:angular-integral} for the moment, we obtain that
\begin{align*}
&\frac{\int_\Omega\abs{ \nabla f_\epsilon}^2 \left( \int_\Omega f_\epsilon^2 \right)^{2/d}}{\int_{\Omega} \abs{f_\epsilon - \overline{ f_\epsilon}}^{2+4/d}} \\
&
 \le \frac{G(d)}{2^{2/d}} \left( 1 - 2 \epsilon C_\kappa \left(\frac{\int_{\R^d} {\abs{\nabla f}}^2 \abs{x} }{\int_{\R^d} {\abs{\nabla f}}^2} + \frac{2}{d} \frac{\int_{\R^d} f^2 \abs{x}} {\int_{\R^d} f^2}    -\frac{\int_{\R^d} f^{2+4/d}\abs{x}} {\int_{\R^d} f^{2+4/d} } \right)  \right)  \\
& \qquad + O(\epsilon^{d}, e^{-R_1/(2 \epsilon)}) .
\end{align*}

In order to compute the sign of the term of order $\epsilon$, we use the Euler-Lagrange equation for $f$. Multiplying \eqref{eq:EL-eq-Rd} by $f$ and integrating gives
\[
\int_{\R^d} \abs{\nabla f}^2 + \int_{\R^d} f^2  = \int_{\R^d} f^{2+4/d}.
\]
Taking the product with $x \cdot \nabla f$ gives after a few integrations by part
\[
(1-d/2 )\int_{\R^d} \abs{\nabla f}^2 - \frac{d}{2} \int_{\R^d} f^2  =  \frac{-d}{2 + 4/d}\int_{\R^d} f^{2+4/d}.
\]
Working out the system finally gives 
\begin{equation} \label{eq:virial_f}
\int_{\R^d} f^{2+4/d} = \frac{d+2}{d} \int_{\R^d} \abs{\nabla f}^2  = \frac{d+2}{2} \int_{\R^d} f^2.
\end{equation}
On the other hand, multiplying the equation by $\abs{x} f$ and integrating gives
\begin{align*}
\int_{\R^d} \abs{x} f^{2+4/d}& -  \int_{\R^d} \abs{x} f^{2} 
= \int_{\R^d} \abs{x} f (- \Delta) f\\
&= \int_{\R^d} \abs{x} \abs{\nabla f}^2  + \frac{1}{2} \int \frac{x}{\abs{x}} \cdot \nabla f^2 \\
&= \int_{\R^d} \abs{x} \abs{\nabla f}^2 + \frac{1}{2}\lim_{r \to 0}\left( -\int_{\partial B_r} f^2 - (d-1)\int_{\R^d \setminus B_r} \frac{f^2}{\abs{x}} \right)\\
&= \int_{\R^d} \abs{x} \abs{\nabla f}^2  - \frac{d-1}{2} \int_{\R^d } \frac{f^2}{\abs{x}}.
\end{align*}
Here, the boundary term vanishes in the limit since $f \in H^1(\R^d)$ implies $f \in L^q(\partial B_r)$ for $q = \frac{2 d }{d-1}>2$.
Inserting this identity together with \eqref{eq:virial_f}, we find that
\begin{align*}
&\frac{\int_{\R^d} {\abs{\nabla f}}^2 \abs{x} }{\int_{\R^d} {\abs{\nabla f}}^2} + \frac{2}{d} \frac{\int_{\R^d} f^2 \abs{x}} {\int_{\R^d} f^2}    -\frac{\int_{\R^d} f^{2+4/d}\abs{x}} {\int_{\R^d} f^{2+4/d} } \\
& = \frac{1}{\int_{\R^d}f^{2+4/d }} \left( \frac{d+2}{d} \int_{\R^d} {\abs{\nabla f}}^2 \abs{x} + \frac{d+2}{2} \,\frac{2}{d} \int_{\R^d} f^2 \abs{x} - \int_{\R^d} f^{2+4/d}\abs{x} \right) \\
& = \frac{1}{\int_{\R^d}f^{2+4/d }} 
\left( \left( \frac{d+2}{d} -1 \right)\int_{\R^d} f^{2+4/d} \abs{x} 
+ \frac{(d+2)(d-1)}{2d}  \int_{\R^d} \frac{f^2}{\abs x}  
\right) \\
& > 0 .
\end{align*}

\begin{figure}
\includegraphics[scale=1]{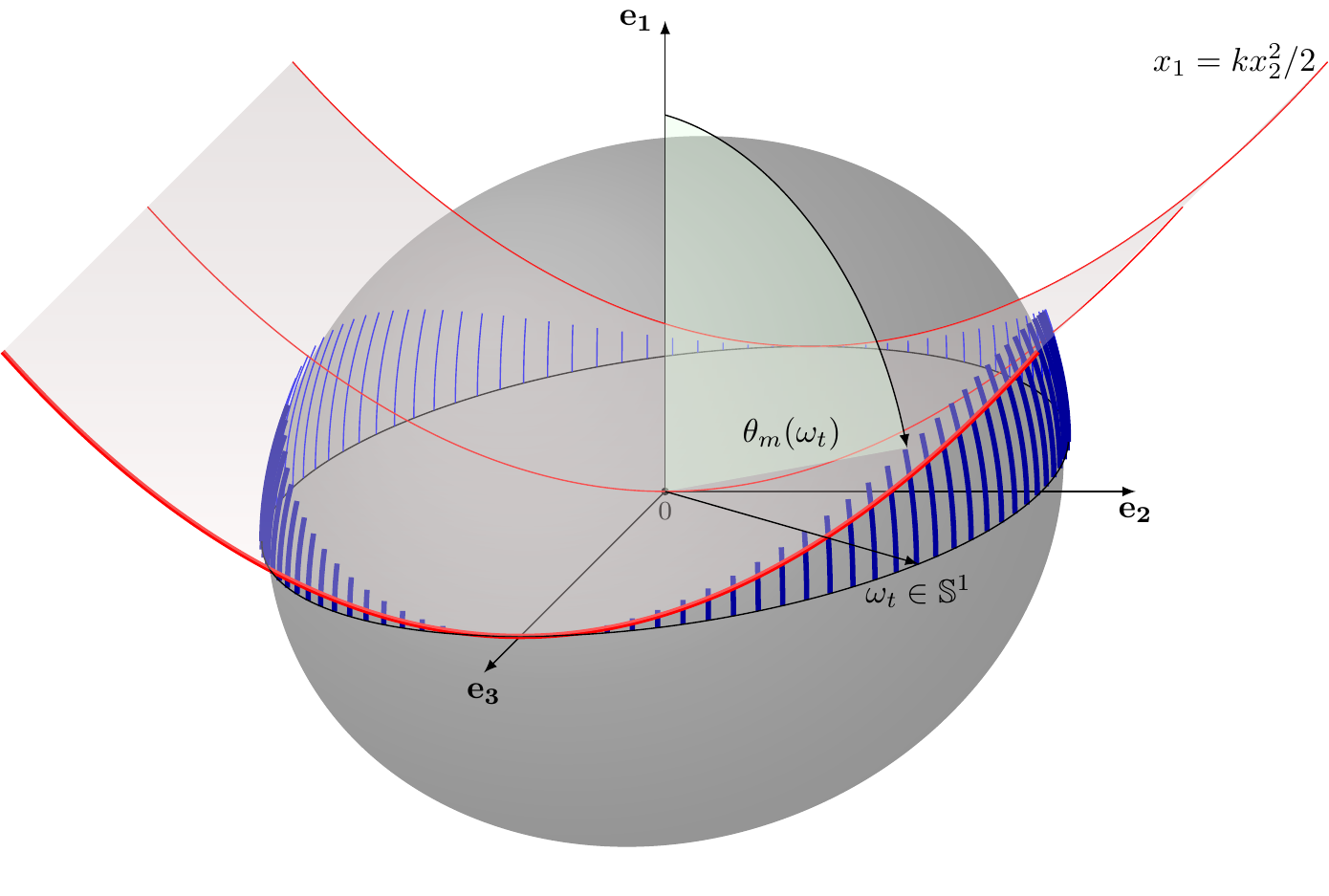}
\caption{Integration over the subset of a spherical shell between the equator and a parabolic surface. In the picture, the curvature matrix ${\bvec K}$ is $\operatorname{diag}(0, k)$, and $r\epsilon = 1$.\label{fig:3d}}
\end{figure}
Now we only have to prove \eqref{eq:angular-integral}. For simplicity of notation we assume that $g$ is supported in $B(0, R_1/ \epsilon)$.
The boundary of $\Omega_\epsilon$ is the graph of
\[
x_1 = \frac{\epsilon}{2} \langle x_t, \bvec{K} x_t\rangle + \widetilde{h}(x_t), \quad \abs{\widetilde{h}(x_t)} \le C \epsilon^2 \abs{x_t}^3.
\] 
We first show that we can replace $\Omega_\epsilon$ by the \emph{parabolic} region
\[
P_\epsilon = \{(x_1, x_t) \in \R^d  | \, x_t \in B(0,R_1/\epsilon), \, x_1 >\frac{\epsilon}{2} \langle x_t, K x_t\rangle \}.
\]
To this end, we define $\widetilde{g}$ on $P_\epsilon$ by
\[
\widetilde{g}(x_1,x_t) = g(x_1 +\widetilde{h}(x_t), x_t).
\]
This change of variables has unit Jacobian. On the other hand,
\begin{align*}
\int_{\Omega_\epsilon }g(x_1, x_t) \di x_1 \di x_t 
&= \int_{P_\epsilon} \widetilde{g}(x_1, x_t) \di x_1 \di x_t \\
&=  \int_{P_\epsilon} \left( g(x_1, x_t) + \int_{0}^{\widetilde h(x_t)}\partial_{1}g(x_1+ s,x_t) \di s \right) \di x_1 \di x_t,
\end{align*}
so 
\[
\left|\int_{\Omega_\epsilon} g -\int_{P_\epsilon} g \right| \le C\epsilon^2 \int_{P_\epsilon} \abs{g'}(x_1, x_t) \abs{x_t}^3 \di x_1 \di x_t.
\]
In spherical coordinates, we have
\begin{align*}
\int_{\R^d_+ \setminus P_\epsilon} g(x) \di x 
= \int_{r \ge 0} g(r) r^{d-1} \int_{\s_+^{d-1} \setminus r^{-1}P_\epsilon}  \di \omega \di r.
\end{align*}
We have to compute the leading order in $\epsilon$ of the angular integral. 
 Introducing (hyper)-spherical coordinates $x= r(\cos(\theta), \sin(\theta) \omega_t)$ with $\omega_t \in \s^{d-2}$, we have to compute the area of $\omega \in \s^{d-1}$ satisfying
\[
\cos(\theta) \le \frac{\epsilon r \sin^2(\theta)}{2} \langle \omega_t, \bvec{K} \omega_t\rangle .
\] 

We start by the easy case $d= 2$, where $\bvec{K} = \kappa >0$. We have to integrate over angles $\theta \in [0, \pi/2]$ satisfying the inequality
\[
\cos(\theta) \le \frac{\epsilon r \kappa}{2} \sin^2(\theta) = \frac{\epsilon r \kappa}{2} (1-\cos^2(\theta))  .
\]
Working out the quadratic equation gives
\[
\cos(\theta_m) = \frac{- 1 + \sqrt{1 + \epsilon^2 r^2 \kappa^2} }{ \epsilon r  \kappa} 
= \frac{\epsilon r  \kappa}{2}  + O(\epsilon r)^2.
\]
So 
\begin{align*}
\int_{\s_+^{2-1} \setminus r^{-1}P_\epsilon}  \di \omega = 2 \int_{\theta_m}^{\pi/2} \di \theta = 2 \frac{\epsilon r  \kappa}{2}  + O(\epsilon r)^2.
\end{align*}
We have obtained \eqref{eq:angular-integral} with $C_2 = 1/(2 \pi)$.

For $d \ge 3$, see Figure~\ref{fig:3d}, the range of $\theta$ depends on $\omega_t$ through an analogous equation and we find that $t_m \equiv \cos(\theta_m)$ is given by
\[
t_m(\omega_t) = \frac{-1 + \sqrt{1+ \bigl(\epsilon r \langle \omega_t, \bvec{K} \omega_t\rangle  \bigr)^2 }}{\epsilon r (\omega_t, \bvec{K} \omega_t)} = \frac{\epsilon r}{2} \langle \omega_t, \bvec{K} \omega_t\rangle + O(\epsilon^2 r^2).
\]
In coordinates with $\bvec{K} = \operatorname{diag}(k_1, \cdots, k_{d-1})$, we find
\begin{align*}
\int_{\s_+^{d-1} \setminus r^{-1}P_\epsilon} \di \omega  &= \int_{\s^{d-2}} \int_{0}^{t_m} (1-t^2)^{(d-3)/2} \di t \di \omega_t 
 \\
   &= \int_{\s^{d-2}} \frac{\epsilon r}{2} \langle \omega_t, \bvec{K} \omega_t\rangle  \di \omega_t + O(\epsilon^2 r^2) \\
&=\frac{\epsilon r}{2} \sum_{i=1}^{d-1} k_i \int_0^\pi \sin(\phi)^{d-3} \cos^2(\phi)  \di \phi \abs{\s^{d-3}} + O(\epsilon^2 r^2)\\
& = \epsilon r \kappa C_d \abs{\s^{d-1}}+ O(\epsilon^2 r^2),
\end{align*}
where $\kappa$ is the mean curvature of $\partial \Omega$ at the origin and 
\begin{align*}
C_d &= \frac{\abs{\s^{d-3}}}{2\abs{\s^{d-1}}} (d-1)  \int_0^\pi \sin(\phi)^{d-3} \cos^2(\phi)  \di \phi \\
& = \frac{(d-2)(d-1)}{4 \pi} B(d/2-1, 3/2) 
\end{align*}
is a constant depending only on the dimension. 
\end{proof}

\section{An example of non-existence} \label{sec:triangle}

Here, we prove that minimizers do not exist in the rectangular isosceles triangle in $\R^2$.
\begin{theorem}
 Let $\Omega \subset{\R^2}$ be the isosceles rectangular triangle. There exist no minimizers for \eqref{eq:var-prob} in $\Omega$
 and 
 \[
G(\Omega, 2) = G(2) /8\sim 0.732.
 \]
\end{theorem}

\begin{proof}
 Let $u$ be a minimizer for \eqref{eq:var-prob} with zero average. We write $u= u_S + u_A$, where $u_S$ and $u_A$ are the symmetric and anti-symmetric parts of $u$ with respect to reflection across the diagonal. 
Note that the anti-symmetric part $u_A$ is zero on the diagonal and has zero average by definition. This means that $u_S$ has zero average as well on each of two isosceles triangles separated by the diagonal.
We claim that
\begin{align*}
 \frac{\int_{\Omega} \abs{\nabla u_A}^2 \int_{\Omega} u_A^2  }{\int_{\Omega} u_A^4  } \ge \frac{G(2)}{4}, \qquad
  \frac{\int_{\Omega} \abs{\nabla u_S}^2 \int_{\Omega} u_S^2  }{\int_{\Omega} u_S^4  } \ge 2 \, G(\Omega, 2). \numberthis \label{eq:ineq-sym-asym}
\end{align*}
Indeed, the first inequality follows from constructing a competitor for the problem in  $\R^2$ from $8$ copies of the restriction of $u_A$ to one of the smaller triangles.
The second inequality follows since $u_S$ has zero average on each of the two triangles separated by the diagonal. These triangles are just scalings of $\Omega$ by $1/\sqrt{2}$. Since $G(\Omega, 2)$ is invariant under dilations of the domain $\Omega$, this gives the inequality. Also note that 
\[
\frac{G(2)}{4} \ge 2 \, G(\Omega, 2)
\]
by Theorem~\ref{thm:treshold-corners} with $\alpha_i = \pi/4$.

We define
\[
\alpha = \frac{\int_{\Omega} \abs{\nabla u_A}^2}{\int_{\Omega}\abs{\nabla u}^2}, \quad
\beta =   \frac{\int_{\Omega} u_A^2}{\int_{\Omega}u^2}, \quad
\gamma=  \frac{\int_{\Omega} u_A^4 + 3\int_{\Omega} u_A^2 u_S^2 }{\int_{\Omega}u^4}.
\]
By definition, these three numbers lie in $[0,1]$ and if one of them is equal to $0$ or $1$, they all are. 

We define a competitor $u_\lambda = u + \lambda u_A$ and compute
\begin{align*}
\frac{d}{d \lambda } \left. \frac{\int_{\Omega} \abs{\nabla u_\lambda}^2 \int_{\Omega} u_\lambda^2  }{\int_{\Omega} u_\lambda^4  } \right|_{\lambda = 1}
= \frac{\int_{\Omega} \abs{\nabla u}^2 \int_{\Omega} u^2  }{\int_{\Omega} u^4  }  
\left(2 \alpha + 2 \beta - 4 \gamma \right).
\end{align*} 
From the minimality of $u$, we obtain $\gamma = (\alpha+ \beta)/ 2 $.
On the other hand, with the definition
\[
\zeta = \frac{\int_{\Omega} u_A^2 u_S^2}{\int_{\Omega}u^4},
\]
we rewrite
\begin{align*}
G(\Omega, 2) &= \frac{\int_{\Omega} \abs{\nabla u}^2 \int_{\Omega} u^2  }{\int_{\Omega} u^4  } \\
 &=
\frac{\gamma - 3 \zeta}{\alpha \beta}\frac{\int_{\Omega} \abs{\nabla u_A}^2 \int_{\Omega} u_A^2  }{\int_{\Omega} u_A^4  } \\
&\ge 
\frac{\gamma - 3 \zeta }{\alpha \beta} 2 G(\Omega,2) \\
& \ge  2 G(\Omega,2) \frac{\frac{\alpha + \beta}{2} - \frac{3}{8}}{ \alpha \beta}, \numberthis\label{eq:CS}
\end{align*} 
where the last line uses $\zeta \le 1/8$, as follows from the Cauchy-Schwarz inequality
\[
2 \int_{\Omega} u_A^2 u_S^2 \le  \int_{\Omega} u_A^4 + \int_{\Omega} u_S^4 = \int_{\Omega} u^4 - 6 \int_{\Omega} u_A^2 u_S^2.
\]
We have found that 
\[
\alpha + \beta - 3/4 \le \alpha \beta. 
\]
Interchanging the roles of $u_S$ and $u_A$, we also obtain that 
\[
(1- \alpha) + (1- \beta) - 3/4 \le (1-\alpha)(1 -\beta ).
 \]
These inequalities imply that $\alpha = \beta = \gamma = 1/2$.
Indeed, the region in the $(\alpha, \beta)$-plane defined by the first inequality touches the diagonal $\alpha + \beta = 1$ only at $\alpha = \beta$, and is otherwise contained in $\alpha + \beta < 1$. The second inequality defines the reflection of the first region across the diagonal $\alpha + \beta = 1$, so the only point of intersection is precisely the center of the unit square.

Having established this, we return to $\eqref{eq:ineq-sym-asym}$.
Since minimizers for the problem in $\R^d$ are not compactly supported, we actually have a strict inequality for $u_A$, so using $\alpha = \beta= 1/2$ in \eqref{eq:CS}, we find
\begin{align*}
G(\Omega,2) =  \frac{1}{2}\frac{\int_{\Omega} \abs{\nabla u_A}^2 \int_{\Omega} u_A^2  }{\int_{\Omega} u_A^4  } >  G(2)/ 8,
\end{align*}
contradicting Theorem~\ref{thm:treshold-corners}.
\end{proof}

\section{Some other existence results.} \label{sec:rectangles}
In this section, we group two results about the existence of minimizers for non-smooth domains which are corollary of Theorem~\ref{thm:treshold-corners} or its analogue for hypercubes, see \cite[Theorem 4.3]{BeVaVdb018}.

\begin{proposition}
Let $\Omega_b \subset \R^2$ be the rectangle $[0,b^{-1}] \times [0,b]$.
There exist $b_c$ satisfying $b_c \le 2.12$ such that, 
For all $b > b_c$, minimizers for \eqref{eq:var-prob} exist. 
\end{proposition}
\begin{proof}
Let $\phi:[0,1] \to \R$ be a smooth function with zero average. 
As a test function for \eqref{eq:var-prob}, we take $u(x_1,x_2)= \phi(x_2/ b)$
and compute
\begin{align*}
\int_{\Omega_b} \abs{u}^p =  \int_0^1 \abs{\phi}^p, \quad \int_{\Omega_b} \abs{\nabla u}^2 =  b^{-2}\int_0^1 (\phi')^2,
\end{align*}
so 
\[
\frac{\int_{\Omega_b} \abs{\nabla u}^2  \int_{\Omega_b} u^2   }{\int_{\Omega_b} {u^4   }} = b^{-2} \frac{\int_0^1 (\phi')^2 \int_0^1 \phi^2}{\int_0^1 \phi^4},
\]
which is smaller than $G(2)/4$ for sufficiently large $b$.
Taking $\phi(x) = \cos(\pi x)$ gives
\[
\frac{\int_0^1 (\phi')^2 \int_0^1 \phi^2}{\int_0^1 \phi^4} = \pi^2 \left(\frac{1}{2}\right)^2 \frac{8}{3} = \frac{2 \pi^2}{3}.
\]
Therefore
\[
G(\Omega_b,2) < G(2)/4 
\]
as soon as 
\[
b > \left(\frac{2 \pi^2/3}{G(2)/4} \right)^{1/2} \sim 2.12\, ,
\]
where we have used the numerical value $G(2) \sim 5.8545$ (see \cite[Table 1]{BeVaVdb018}).
Numerically optimizing over anti-symmetric functions gives that there exists $\phi$ with zero average satisfying
\[
\frac{\int_0^1 ({\phi'})^2 \int_0^1 \phi^2}{\int_0^1 \phi^4} \le 6.1622 \, ,
\]
which gives a slightly better bound for $b_c$.
\end{proof}

\begin{proposition}
Let $\Omega_d$ be the $d$-dimensional hypercube. Minimizers for \eqref{eq:var-prob} exist for $d \ge 10$.
\end{proposition}
\begin{proof}
Again, we take as a test function $u(x) = \cos(\pi x_1)$ and compute
\begin{align*}
T(d)=\frac{\int_{\Omega_d} \abs{\nabla u}^2  \left( \int_{\Omega_d} u^2 \right)^{2/d}   }{\int_{\Omega_d} {\abs{u}^{2+4/d}   }} 
&= \pi^2 \frac{\left( \int_0^1 \cos^2(\pi x) \di x \right)^{1+2/d}}{\int_0^1 \left(\cos^2(\pi x)\right)^{1+2/d} \di x} \\
&=\frac{\pi^2}{2^{1+2/d}}\, B(3/2+2/d\, , 1/2)^{-1}  
\end{align*}
where $B(x,y)=\Gamma(x)\, \Gamma(y)/\, \Gamma(x+y)$ is the Euler Beta function. From the first expression, we also see that $T(d) \le \pi^2$ for all $d$. On the other hand, for $d \ge 3$, we have (see e.g., \cite{Bec004} or \cite[Section 4.4]{Lun017})
\[
G(d) \ge S_d = \frac{d(d-2)}{4} \abs{\s^d}^{2/d},
\]
where $S_d$ is the sharp constant in the Sobolev inequality.

Thus, $G(\Omega,d) < G(d)/4$ if $d$ is such that
\[
T(d) \, < \frac{d(d-2)}{16} \abs{\s^d}^{2/d}\, .
\]
Since the right hand side of this inequality grows linearly for large $d$, this is always satisfied for $d$ sufficiently large.
Explicitly, we check that
\[
\frac{11 \times 9}{16}  \abs{\s^{11}}^{2/11} = 10.246 > 9.293 =  T(11) .
\]
Moreover, numerically we have (see \cite[Table 1]{BeVaVdb018}) that $T(10)\sim 9.233 < G(10)/\, 4\sim 9.536$, so minimizers exist for $d=10$.
For $d \ge 11$, we simply use the fact that $S_d$ is increasing with $d$ and
\[
T(d) \le \pi^2 < S_{11} \le S_d.
\]
\end{proof}
\appendix

\section{Criticality of $p= 2/d$}
In this appendix, we consider the generalized problem 
\begin{equation}
\label{eq:var_prob_p}
G_p(\Omega, d) = \inf_{H^1(\Omega)} \frac{\int_\Omega \abs{\nabla u}^2 \left(\int_\Omega u^2 \right)^{p}}{\int_\Omega \abs{u-u_\Omega}^{2+2p}}, \quad 0 \le p \le 2/(d-2) .
\end{equation}

The problem for $p = 2/d$ is critical in the following sense.

\begin{theorem}
Let $\Omega \subset \R^d$ be a bounded domain with locally Lipschitz boundary. For all $p\in [0,2/d)$ a minimizer for $G_p(\Omega, d)$ exists. For all $p \in (2/d, 2/(d-2)]$, $G_p(\Omega, d) =0$ and no minimizer exists.
\end{theorem}

\begin{proof}
 {\bf Non-existence for supercritical $p$.}
 By scaling. Take the origin in the interior of $\Omega$, $\phi \in C_c^{\infty}$ with zero average. For sufficiently large $\lambda>0$, the support of $u_\lambda \equiv \phi(\lambda \cdot)$ is in $\Omega$.
 We compute
\begin{align*}
 \frac{\int_\Omega \abs{\nabla u_\lambda}^2 \left(\int_\Omega u_\lambda^2 \right)^{p}}{\int_\Omega \abs{u_\lambda}^{2+2p}}
 = \frac{\lambda^{2-d}\int_{\R^d} \abs{\nabla \phi}^2 \left(\lambda^{-d}\int_{\R^d} \phi^2 \right)^{p}}{\lambda^{-d}\int_{\R^d} \abs{\phi}^{2+2p}} = \lambda^{2-pd} C_\phi.
\end{align*}
This tends to zero when $\lambda$ increases if $p>2/d$.

\noindent {\bf Existence for sub-critical $p$.}
In this case, we prove existence of a minimizer sequence by showing that minimizing sequences can not concentrate in small sets. As in the proof of Theorem~\ref{thm:treshold-smooth}, if a minimizing sequence has no convergent subsequence, we obtain from Lemma~\ref{lem:localizing},
\[
G_p(\Omega, d) \ge \liminf_{\delta \to 0}  G_p(\Omega, d,\delta).
\]
Now, we show that 
\begin{equation} \label{eq:explosion}
 G_p(\Omega, d,\delta) \ge C_{p, \Omega} \delta^{d p -2},
\end{equation}
with some constant $C_{p, \Omega}>0$ depending only on $p$, the dimension $d$ and the Lipschitz constant of $\Omega$.

First of all, by using H\"older's inequality for $f= \abs{u}^{2+2p}$, $g= 1$ in the denominator,
\begin{align*}
\inf_{u \in H^1_0(B(0,1)) } \frac{\int \abs{\nabla u}^2 \left(\int  u^2 \right)^{p}}{\int  \abs{u }^{2+2p}}\ge  \left(\frac{\mu_1}{\omega_d} \right)^{\frac{2-pd}{2+d}}G(d)^{\frac{1+p}{1+2/d}}>0,
\end{align*}
with $\mu_1$ the first Dirichlet eigenvalue of $B(0,1)$ and $\omega_d$ the volume of $B(0,1)$.

By scaling, if $B(s,\delta) \subset \Omega$,
\begin{align*}
\inf_{u \in H^1_0(B(s,\delta)) } \frac{\int_\Omega \abs{\nabla u}^2 \left(\int_\Omega  u^2 \right)^{p}}{\int  \abs{u_\Omega }^{2+2p}}\ge  \delta^{d p -2} \left(\frac{\mu_1}{\omega_d} \right)^{\frac{2-pd}{2+d}}G(d)^{\frac{1+p}{1+2/d}}.
\end{align*}

Otherwise, by replacing $\delta$ with $2 \delta$, we may assume that $s$ is on the boundary.
In this case we assume that $\delta$ is small enough such that $B(s,2\delta) \cap \partial \Omega$ is the graph of a Lipschitz function over some hyperplane passing through $s$. We chose coordinates such that this hyperplane coincides with $x_1 =0$.
and write 
\[
(x_1, x_t) \in \Omega \cap B(s,2\delta) \Leftrightarrow (x_1,x_t) \in B(s, 2 \delta) \text{ and } x_1 > h(x_t),
\]
for some Lipschitz function $h$.
We define $f$ with support in $B(s, 2 \delta) \cap \R^d_+$ by
$f(x_1, x_t) = v(x_1 + h(x_t), x_t)$.
As before, this change of variables has unit Jacobian and, by definition of a Lipschitz domain \cite[Chapter IV]{Ada975}, the distributional derivative of $h$ is bounded by the Lipschitz constant $L$. Therefore,
\begin{align*}
 \frac{\int_\Omega \abs{\nabla v}^2 \left(\int_\Omega v^2 \right)^{p}}{\int_\Omega \abs{v}^{2+2p}} 
 &\ge \frac{1}{(2+2L)^2} \frac{\int_{\R^d_+} \abs{\nabla f}^2 \left(\int_{\R^d_+} f^2 \right)^{p}}{\int_{\R^d_+} \abs{f}^{2+2p}} \\
 &\ge \frac{1}{(2+2L)^2} 2^{-p} \delta^{d p -2}\left(\frac{\mu_1}{\omega_d} \right)^{\frac{2-pd}{2+d}}G(d)^{\frac{1+p}{1+2/d}}.
\end{align*}
This proves \eqref{eq:explosion}.
Thus, if a minimizing sequence does not have a convergent subsequence,
\[
G_p(\Omega, d) \ge \liminf_{\delta \to 0}  G_p(\Omega, d,\delta) = + \infty,
\]
which is clearly a contradiction.
\end{proof}

\bigskip
\bigskip

\section*{Acknowledgments}
\thanks{The work of R.B. and C.V. has been supported by Fondecyt (Chile) Project \# 116--0856. The work of H. VDB. has been partially supported by  CONICYT (Chile)  (PCI) project REDI170157 and partially by Fondecyt (Chile) Project \# 318--0059


\end{document}